\documentclass{llncs}

\usepackage{amssymb}			% useful mathematical symbols
\usepackage{amsmath}	% "cases" environment

\usepackage{tikz} 
\usetikzlibrary{chains,positioning,scopes,arrows,plotmarks}	% drawing trees
\usepackage{comment}								% comentarios
\usepackage{array}								% thick vertical lines in tables
\usepackage{tabto}								% tabulaciones
\usepackage{xspace}
\usepackage{color}
\usepackage[utf8]{inputenc}
\usepackage{psfrag}

\usepackage[toc,page]{appendix}

\usepackage{array}
\usepackage{booktabs}
\usepackage{multirow}% http://ctan.org/pkg/multirow
\usepackage{hhline}% http://ctan.org/pkg/hhline
\usepackage{graphics}

\newcommand{\NE}{\textsc{ne}\xspace}

\newcommand{\remove}[1]{}

\author{C. \`Alvarez \and A. Messegu\'e}
\institute{ALBCOM Research Group, Computer Science Department, UPC, Barcelona\\
\email{\{alvarez,messegue\}@cs.upc.edu}}

\title{Network Creation Games: Structure vs Anarchy}%\thanks{partially supported by MINECO grant TIN2013-46181-C02-1-R (COMMAS) and AGAUR grant SGR 2014:1137 (ALBCOM).}}

\begin{document}

\maketitle
\begin{abstract}

We study Nash equilibria and the price of anarchy in  the classical model of Network Creation Games introduced by Fabrikant et al. In this model every agent (node) buys links at a prefixed price $\alpha>0$ in order to get connected to the network formed by all the $n$ agents.
In this setting, the  reformulated tree conjecture  states that for $\alpha > n$, every Nash equilibrium network is a tree. Since  it was shown that the price of anarchy for trees is constant, if the tree conjecture were true, then the price of anarchy would be constant for $\alpha >n$. 
Moreover, Demaine et al.  conjectured  that  the price of anarchy for this model is constant. 
 
Up to now the last conjecture has been proven in (i) the \emph{lower range},  for $\alpha = O(n^{1-\epsilon})$ with $\epsilon \geq \frac{1}{\log n}$ and (ii) 
in the \emph{upper range},  for $\alpha > 65n$.
In contrast,  the best upper bound known for the price of anarchy for  the remaining  range  is $2^{O(\sqrt{\log n})}$. 

In this paper we give new insights into the structure of the Nash equilibria for different ranges of $\alpha$ and we enlarge the  range  for which the price of anarchy is  constant. 
Regarding the upper range,  we prove  that every Nash equilibrium is a tree for $\alpha > 17n$ and that the price of anarchy is constant even for $\alpha > 9n$. In the lower range, we show that any Nash equilibrium for $\alpha < n/C$ with $C > 4$, induces an $\epsilon-$distance-almost-uniform graph. 

\end{abstract} 
\section{Introduction}
\label{sec:intro}

This article focuses its attention on the \emph{sum classic network creation game}  introduced by Fabrikant et al. in \cite{Fe:03}.
This strategic game models Internet-like networks without central coordination. In this model the distinct agents, who can be thought as nodes in a graph,  establish links of constant price $\alpha$ to the other agents in order to be connected in the resulting network.  We analyze the structure of   the resulting equilibrium networks as well as their performance under the price of anarchy. Hence, our main elements of interest are \emph{Nash equilibria} (\NE), configurations where every agent is not interested in deviating his current strategy, and the \emph{price of anarchy} ($PoA$), a measure of how the efficiency of the system degrades due to selfish behaviour of its agents. 

 \textbf{Related work.} In the seminal article from Fabrikant et al. \cite{Fe:03} it was shown that the $PoA$ of sum classic network creation games  is $O(\sqrt{\alpha})$. In the subsequent years the range $\alpha$ for which the $PoA$ is constant has been  enlarged. Table 1 contains a summary of the best upper bounds on the $PoA$ for the different values of the parameter $\alpha$. 
 %We can distinguish between three ranges for $\alpha$: the lower, the middle and the upper range.
 For the lower range, it was  proved that the $PoA$ is constant for $\alpha = O(\sqrt{n})$ in \cite{Lin} and \cite{Albersetal:06}, independently. Afterwards, this range was enlarged  in \cite{Demaineetal:07} by showing that the $PoA$ is constant for $\alpha = O(n^{1-\epsilon})$, with $\epsilon \geq 1/\log n$. Futhermore, in \cite{Demaineetal:07}, the authors  also provided better upper bounds  for  $\alpha < \sqrt[3]{n/2}$ and for $\alpha < \sqrt{n/2}$.  For the upper range, it was first proved that the $PoA$ is constant for  $\alpha = \Omega(n^{3/2})$ in \cite{Lin} and later,  a constant upper bound on the $PoA$ was  also shown  for  $\alpha \geq 12n  \log n $ in \cite{Albersetal:06}.  Subsequently, it was proved that any \NE is a tree, first for the range $\alpha > 273n$ in \cite{Mihalakmostly} and more recently for the range $\alpha > 65n$ in \cite{Mihalaktree}. Hence,  the $PoA$ is constant for $\alpha > 65n$.  For the remaining range, it was first proven in \cite{Albersetal:06} an upper bound on the $PoA$ of $15\left(1 + \min \left\{ \alpha^2/n, n^2/\alpha\right\}^{1/3}\right)$ for $\alpha < n$ . Later on, this result was improved to $2^{O(\sqrt{\log n})}$ for $\alpha < 12n \log n$ in \cite{Demaineetal:07}.

%\vskip 5pt
  \begin{center}
  \resizebox{\columnwidth}{!}{
  \begin{tabular}{c|c|c|c|c|c|c|c|c|c|c|c|c|c|c|c|c|c|c|c|c|}
  \multicolumn{2}{c}{\hspace{0.0cm}$\alpha = 0$ } & \multicolumn{2}{c}{\hspace{0.0cm}1}  & \multicolumn{2}{c}{\hspace{1.2cm}2}  & \multicolumn{2}{c}{\hspace{0.7cm}$\sqrt[3]{n/2}$}  & \multicolumn{2}{c}{\hspace{0.5cm}$\sqrt{n/2}$} & \multicolumn{2}{c}{\hspace{0.5cm}$O(n^{1-\epsilon})$} & \multicolumn{2}{c}{\hspace{1.45cm}$9n$}  & \multicolumn{2}{c}{\hspace{1.8cm}$17n$} & \multicolumn{2}{c}{\hspace{1.5cm}$65n$} & \multicolumn{2}{c}{\hspace{0.4cm}$12n \log n$}  &  \multicolumn{1}{c}{\hspace{0.5cm}$\infty$}\\

  &  \multicolumn{2}{c|}{}  &  \multicolumn{2}{c|}{}  & \multicolumn{2}{|c|}{}  & \multicolumn{2}{|c|}{} & \multicolumn{2}{|c|}{} & \multicolumn{2}{|c|}{}  & \multicolumn{2}{|c|}{} & \multicolumn{2}{|c|}{}  & \multicolumn{2}{|c|}{} & \multicolumn{2}{|c|}{}    \\
\cline{2-21}
  
 $PoA$ & \multicolumn{2}{|c|}{1}  &  \multicolumn{2}{|c|}{$\leq \frac{4}{3}$ (\cite{Fe:03})}  & \multicolumn{2}{|c|}{$\leq 4$ (\cite{Demaineetal:07})}  & \multicolumn{2}{|c|}{$\leq 6$ (\cite{Demaineetal:07})} & \multicolumn{2}{|c|}{$\Theta(1)$ (\cite{Demaineetal:07})} & \multicolumn{2}{|c|}{$2^{O(\sqrt{\log n})}$ (\cite{Demaineetal:07})}  & \multicolumn{2}{|c|}{$\Theta(1)$ (Thm. \ref{thm_9n}) }  & \multicolumn{2}{|c|}{$<5$ (Thm. \ref{thm_tree})} & \multicolumn{2}{|c|}{$<5$ (\cite{Mihalaktree})} & \multicolumn{2}{|c|}{$1.5$ (\cite{Albersetal:06}) }     \\
  \cline{2-21}
  \end{tabular}
  }

  \vskip 5pt
 \noindent {\small \textbf{Table 1.} Summary of the best known bounds for the $PoA$ for the sum classic game.}
 
  \end{center}
  
In all these previous results one can see that analyzing the structure of the resulting \NE graphs is very useful for shedding light on the $PoA$.  In particular, in \cite{Fe:03} it was shown that any tree Nash equilibrium is less than 5 times more costly than the social optimum.
In \cite{Fe:03}  it was conjectured that there is a constant $A$ for which every \NE is a tree whenever $\alpha > A$. This \emph{tree conjecture}  was refuted later in \cite{Albersetal:06}. Nevertheless, if the condition $\alpha > A$ is relaxed to the condition $\alpha > f(n)$, the tree conjecture can be formulated. As we have pointed before,  the tree conjecture is known to be true   for $\alpha > 65n$. In order to show this result,  the auhors   in  \cite{Mihalaktree} provided new  insights on the local structure of the \NE, in particular on the length of the network's shortest cycle and on the average degree. The tree conjecture is expected to be true for $\alpha > n$.  

Furthermore, in \cite{Demaineetal:07} it was shown that for $\alpha \geq 2$ the $PoA$ of a \NE graph $G$ is upper bounded by  $diam(G)+1$.
Hence, related with the structure of the \NE graphs, a fundamental key question is how well they globally minimize the diameter in order to have a low $PoA$.
 
 \textbf{Our results.} We show new  local properties on the structure of the Nash equilibria that allow us to enlarge the  range  of $\alpha$ for which the price of anarchy is  constant as well as the range of $\alpha$  for which every \NE is a tree. 

In section 3, by obtaining a lower bound on  the average degree of a $2-$edge connected component of a \NE graph, we can prove that that every \NE  is a tree
for $\alpha > 17n$. Furthermore, we show an upper bound for the diameter of a \NE graph in terms of the diameter of its $2-$edge connected components.  This result jointly with an improvement on  the lower  bound on the average degree of the $2-$edge connected components of  \NE lead us to  show that the $PoA$ is constant even for $\alpha > 9n$.

In section 4 we show that the fourth power of any \NE  is  an $\epsilon-$distance-almost-uniform graph (for an appropriate $\epsilon > 0$)  for the range $\alpha < n/C$ with $C>4$. Distance-almost-uniform graphs were introduced  by Alon et al. in 
 \cite{AlonDHKL14} where it was  conjectured that distance uniform  graphs have diameter $O(\log n)$. 
Recently  in \cite{LaLo:17},  M. Lavrov and Lo refute  such conjecture.

%\textbf{Structure of the article.} In Section 2 we introduce the basic definitions and the model. In Section 3 we prove that every \NE for $\alpha > 17n$ is a tree and that the $PoA$ for $\alpha > 9n$ is constant. Finally, in Section 4 we study the case $\alpha < n/C$ with $C > 4$ a constant.

\section{Preliminaries}
\label{sec:model}
 \textbf{The model.} A \emph{network creation game} is defined by a set of players $V = \left\{1,2,....,n \right\}$ and a positive parameter $\alpha$. Each player $u$ represents a node of an undirected graph and $\alpha$ the cost per link. The \emph{strategy} of a player $u \in V$ is denoted by $s_u$ and is a subset $s_u \subseteq V \setminus \left\{u \right\}$ which represents the set of nodes to which player $u$ wants to be connected. The strategies of all players define the \emph{strategy vector} $s=(s_{u})_{u \in V}$. The \emph{communication network} associated to a strategy vector $s$ is then  defined as the undirected graph $G_s = (V, \left\{uv \mid v \in s_u \lor u \in s_v \right\})$, which is the natural network formed by the choices of the players. For the sake of convenience $G_s$ can be understood as directed or undirected at the same time. On the one hand, we consider the directed version when we are interested in the strategies of the players defining the communication graph. On the other hand, we focus on the undirected version when we want to study the properties of the topology of the communication graph. The cost associated to a player $u \in V$ is $c_u(s) =  \alpha |s_u| + \sum_{v \neq u} d_{G_s}(u,v)$, where $d_{G_s}(u,v)$ is the distance between $u$ and $v$ in $G_s$. Thus, the social cost $c(s)$ of the strategy vector $s$ is defined by the sum of the individual costs, i.e. $c(s)= \sum_{u \in V}{c_u(s)}$. A \emph{Nash Equilibrium} ( \NE) is a strategy vector $s$ such that  for every player $u$ and 
 every strategy vector $s'$ differing from $s$ in only $u$, $c_u(s) \leq c_u(s')$. In a \NE no player has incentive to deviate individually his strategy. Finally, let $E$ be the set of \NE. The \emph{price of anarchy} is the ratio $PoA= \max_{s \in E}c(s)/\min_{s}c(s)$. $PoA$ is understood as a measure of how the efficiency of the system may be degraded due to selfish behaviour of the agents.

% Notice that there is a bijection between the communication networks and the strategy vectors. More precisely, given a digraph $G = (V,E)$ we can define $\bar{s}_u(G) = \left\{ v \in V \mid (u,v) \in E \right\}$ for any player $u \in V$ and $\bar{s}(G) = (\bar{s}_u(G))_{u \in V}$ in such a way that $s = \bar{s}(G_s)$ and $G = G_{\bar{s}(G)}$. Since we will be working mainly with digraphs rather than strategy vectors we now extend the previous definitions to digraphs: we say that a digraph $G = (V,E)$ is a \NE iff $\bar{s}(G)$ is a \NE with respect the previous definitions. Also, for any player $u\in V$ we set $c_u(G) = c_u(\bar{s}(G))$ and $c(G) = \sum_{u \in V} c_u(G)$. In this way the Price of Anarchy can be then redefined as $\max_{G\in \bar{E}}c(G)/\min_{G}c(G)$, where $\bar{E}$ is the set of \NEs with $V$ as the vertex set. 

%\vskip 5pt

 \textbf{Graphs.} In a digraph $G$ the edges are considered to have an orientation and $(u,v)$ denotes an edge from $u$ to $v$. In contrast, for an undirected graph $G$, the edge from $u$ to $v$ is the same as the edge from $v$ to $u$ and it is denoted as $uv$.
Given a digraph $G = (V,E)$, a node $v \in V$  and $X \subseteq G$ a subgraph of $G$ let $deg_X^+(v) = | \left\{ u \in V(X) \mid (v, u) \in E\right\}| $, $deg_X^-(v) =  | \left\{ u \in V(X) \mid  (u,v) \in E\right\}|$ and $deg_X(v) = deg_X^+(v)+deg_X^-(v)$. Likewise, if $G = (V,E)$ is an undirected graph and $v \in V$ any node we define $deg_X(v) = |\left\{ u \in V(X) \mid  uv \in E \right\}|$. If $X = G$ then we drop the reference to $G$ and write $deg^+(v),deg^-(v), deg(v)$ instead of $deg_G^+(v),deg_G^-(v), deg_G(v)$.
 
 In a connected graph $G=(V,E)$ an edge $e \in E$ is a \emph{bridge} if its removal increases the number of connected components of $G$. A graph is \emph{$2-$edge connected} if it has no bridges. We say that $H \subseteq G$ is a \emph{$2-$edge connected component} of $G$  if $H$ is a maximal $2-$edge connected subgraph of $G$. In this way, for any $u \in V(H)$ we define $T(u)$ as the connected component containing $u$ of the subgraph induced by the vertices $(V(G)\setminus V(H)) \cup \left\{u \right\}$. The weight of a node $u \in V(H)$ is then defined as $|T(u)|$.

 A \emph{minimal cycle} $C$ in $G$ is a cycle that cannot be shortened, i.e, a cycle such that $d_G(u,v)=d_C(u,v)$ for every two nodes $u,v\in C$, where $d_C(u,v)$ is the minimum number of consecutive edges of $C$ needed to go from $u$ to $v$. A cycle $C = u_0-u_1-u_2-...-u_{k-1}-u_0$ is \emph{directed} if either $(u_i,u_{i+1}) \in E(G)$ for each $i$ mod $k$ or if $(u_{i+1},u_i) \in E(G)$ for each $i$ mod $k$. Finally, we define $g(G)$ to be the girth of $G$, that is, $g(G) = \infty$ if $G$ is a tree, otherwise, $g(G)$ is the minimum length that any cycle in $G$ can have.

\section{The upper range}
\label{sec:poa-upper-bound}
 In this section we prove that for $\alpha > 17n$ every \NE is a tree and that for $\alpha > 9n$ the $PoA$ is constant. In order to do so, let $G$ be a \NE and $H \subseteq G$ a non-trivial $2-$edge connected component, i.e, a $2-$edge connected component having at least $3$ nodes. Our reasoning goes as follows:

 First, in subsection \ref{subsec:directed-cycles-and-forbidden-topologies} we pay attention at the nodes of degree exactly two in $H$ and we see that, for  $\alpha > 4n$, every node of this kind has  bought exactly one link (we call these nodes $2-$nodes). After, we consider the paths in $H$ consisting only in consecutive $2-$nodes (we call these paths $2-$paths), and we see that the maximum length of a $2-$path is $4$ whenever $g(G) > 14$ (see  subsection \ref{subsec:coordinates-and-2-paths}). After this, in subsection \ref{subsec:an-imprpovement-towards-the-tree-conjecture}, we show that the subgraph induced by maximal $2-$paths must be a forest whenever $\alpha > 4n$ and $g(G)>14$. 
 Gathering together all these results we obtain a better lower bound for the average degree of $H$.
 %This result jointly with the previous ones allow us to get a better lower bound for the average degree of $H$.   
 Combining this bound with the upper bound given in \cite{Mihalaktree}, we deduce that every \NE is a tree for $\alpha > 17n$. Furthermore, if we require $H$ to have at least a certain diameter, then the range of the parameter $\alpha$ for which the $PoA$ is constant is enlarged from $\alpha > 65n$ to $\alpha > 9n$ as we show in subsection \ref{subsec:a-further-improvement-on-the-price-of-anarchy}.

\subsection{Directed cycles and forbidden topologies}
\label{subsec:directed-cycles-and-forbidden-topologies} 

 A node $u\in V(H)$ is a \emph{$2-$node} if $deg_H^-(u) = deg_H^+(u) = 1$. Let us see  that every node of degree $2$ in $H$ is a $2-$node and 
hence we can not have the following topologies (looked in $H$): 
 
\begin{tikzpicture}
\tikzset{vertex/.style = {shape=circle,draw,minimum size=1.5em,scale=0.75}}
\tikzset{edge/.style = {->,> = latex'}}
% vertices
\node[vertex] (u) at  (0,0) {$u$};
\node[vertex] (w) at  (2,0) {$w$};
\node[vertex] (v) at  (4,0) {$v$};
%edges
\draw[edge] (u) to (w);
\draw[edge] (v) to (w);

\node[vertex] (u1) at  (6,0) {$u$};
\node[vertex] (w1) at  (8,0) {$w$};
\node[vertex] (v1) at  (10,0) {$v$};
%edges
\draw[edge] (w1) to (u1);
\draw[edge] (w1) to (v1);

\end{tikzpicture}

  Let $D(u) = \sum_{v \neq u}d_G(u,v)$ be the distance component of the cost function. In any \NE, whenever a node $u$ swaps a link $(u,v_1)$ for the link $(u,v_2)$, if the outcome graph remains connected, then $u$ has the option to use the link $(u,v_2)$ to reach any other node $v$ at a distance no greater than the distance between $v_2$ and $v$ plus one unit. Using this property, we show that the distance component of the cost function does not vary too much  among distinct nodes in $H$.

\begin{proposition} \label{propo_diff}
If $G$ is a \NE and $H\subseteq G$ a non-trivial $2-$edge connected component then $|D(u)-D(v)| \leq 3n$ for all $u,v\in V(H)$. 
\end{proposition}

\begin{proof}

 Let $u_0 \in V(H)$ be a node minimizing the function $D(\cdot)$ over all nodes in $H$. Let $X_i = \left\{ v \in V(H) \mid D(v) \leq in + D(u_0) \right\}$ for nonnegative integers $i$. Also, define the sets $A_{H,r}(u)$ to be the set of nodes in $H$ at distance $r$ from $u$ and the number $n_{H,u_0}^+(u)$ as the number of nodes $v \in V(H)$ connected to $u$ (either $(u,v)$ or $(v,u)$ in $E(H)$) such that $d_G(u_0,v) = d_G(u_0,u)+1$. Consider the following cases. First, let $u$ be such that $n_{H,u_0}^+(u) = 0$ and suppose that $u$ has bought at least one link $(u,v)$. Then deviating by deleting $(u,v)$ and adding a link to $u_0$ he gets a cost difference of at most $n-1+D(u_0)-D(u)$ so that since $G$ is a \NE we must have $D(u) < n + D(u_0)$, that is, $u \in X_1$. Now assume that $n_{H,u_0}^+(u) = 0$ but $u$ has not bought any link. Since $H$ is connected there must exist at least one node $v$ for which $v$ has bought the link $(v,u)$. Since $H$ is $2-$edge connected notice that if $v$ deviates deleting the link $(v,u)$ and buying a link to $u_0$ then $v$ gets a cost difference of at most $n-1 + D(u_0)-D(v)$. Again, since $G$ is a \NE we must have $v \in X_1$ and therefore $u \in X_2$. Finally, let $u \in V(H) \cap A_{H,r}(u_0)$ be a node with $n_{H,u_0}^+(u) > 0$. If $u$ has bought the link $(u,v)$ with $v \in A_{H,r-1}(u_0) \cup A_{H,r}(u_0)$ then $u$ can deviate deleting the link $(u,v)$ and buying the link to $u_0$ getting a cost difference of at most $n-1 + D(u_0)- D(u)$. Then we have that $u \in X_1$. If $v \in A_{r+1,H}(u_0)$ then we can build a path starting at $v$ going each time one step further away from $u_0$ until we reach a node $w$ satisfying $n_{H,u_0}^+(w)=0$. If $w$ has bought at least one link then by similar arguments we can show that $w\in X_1$ so that $u \in X_2$. Whereas if $w$ has not bought any link then consider the predecessor $w'$ of $w$ in the path we have followed. Since $w' \in X_1$ by previous results we deduce that $v\in X_2$. The last case is when $n_{H,u_0}^+(u) > 0$ and $u$ has not bought any link. In this case notice that $u$ must be adjacent to a node $v$ satisfying either $n_{H,u_0}^+(v) > 0$ and then $v \in X_2$ and then $u \in X_3$ or $n_{H,u_0}^+(v) = 0$ so that $v \in X_2$ and then $u \in X_3$, too. Hence, $u \in X_3$ for all $u \in V(H)$. 

\end{proof}

 Moreover, notice that in any non-directed cycle $C$ there exists a node $u$ such that $deg_C^+(u) = 2$. If such cycle is minimal and $\alpha > 4n$ then node $u$ will have incentive to deviate his strategy by removing these two links and buying a link to its furthest node in $C$.

\begin{proposition} \label{propo_cicles2}
Let $G$ be a \NE for $\alpha > 4n$ and $H \subseteq G$ a non-trivial $2-$edge connected component. Then, every minimal cycle in $H$ is directed. 
\end{proposition}

\begin{proof}

 Let $C = u_0-u_1-...-u_{k-1}-u_0$ be a minimal cycle in $H$. Assume the contrary, then there exists at least one node $u_i \in C$ such that it has bought two links: $(u_i,u_{i+1})$ and $(u_i,u_{i-1})$. Assume wlog that $i=0$. If $k = 2k'$ with $k' \in \mathbb{N}$ then it is clear that $d_G(u_{k-1},u_{k'}) = d_G(u_{1},u_{k'}) = d_G(u_0,u_{k'})-1$. In this case the deviation that consists in deleting the links $(u_0,u_{k-1}), (u_0,u_{1})$ and buying the link $(u_0,u_{k'})$ has a cost difference of at most $-\alpha + n + D(u_{k'})-D(u_0) \leq -\alpha + n + 3n < 0$ for $\alpha > 4n$. A similar argument can be used to show the same result when $k = 2k'+1$ with $k' \in \mathbb{N}$: in this case $d_G(u_{k-1},u_{k'}) = d_G(u_1,u_{k'})+1 = d_G(u_0,u_{k'})$ so that if $u_0$ deletes the links $(u_0,u_{k-1}),(u_0,u_1)$ and buys a link to $u_{k'}$ the corresponding cost difference is at most $-\alpha + n + D(u_{k'})-D(u_0) \leq -\alpha + n + 3n < 0 $ for $\alpha > 4n$ and, again, a contradiction is reached. 

\end{proof}

 Furthermore, it is no hard to see that every edge of a $2-$edge connected graph is contained in a minimal cycle.

\begin{proposition} \label{propo_cicles1} 
Let $H$ be a $2-$edge connected graph. Then, for every $(u,v) \in E(H)$ there is a minimal cycle $C$ containing $(u,v)$.
\end{proposition}

\begin{proof}

 Since $H$ is $2-$edge connected then there exists at least one path $u = v_0-v_1-v_2-...-v_k=v$ not containing the edge $(u,v)$. Now consider $\Gamma$ the set of cycles containing the edge $(u,v)$. We have that $\Gamma \neq \emptyset$ because the closed path $v_0-v_1-...-v_k-v_0$ is a cycle containing $(u,v)$. Let $C \in \Gamma$ be a cycle with minimum perimeter among all the cycles in $\Gamma$. If $C = u_0-u_1-...-u_l-u_0$ was not a minimal cycle, then there would exist at least two non consecutive subindexes $i,j$ such that $d_G(u_i,u_j) < d_C(u_i,u_j)$. In this case, considering the minimal length path between $u_i$ and $u_j$ a cycle still containing the edge $(u,v)$ but with less perimeter than $C$ could be obtained considering the minimal length path between $u_i$ and $u_j$, a contradiction. 

\end{proof}

 Taking into account Propositions \ref{propo_cicles2} and \ref{propo_cicles1} we notice that every node $u \in V(H)$ satisfies that $deg_H^+(u) \geq 1$. Following the main idea behind the proof of Proposition \ref{propo_diff} we have the following two results.
\begin{corollary}\label{corol_diff_cost}
Let $G$ be a \NE for $\alpha > 4n$ and $H\subseteq G$ a non-trivial $2-$edge connected component. Then $|D(u)-D(v)| < 2n$ for any two nodes $u,v\in V(H)$.

\end{corollary}
\begin{corollary} Let $G$ be a \NE for $\alpha > 4n$ and $H\subseteq G$ a non-trivial $2-$edge connected component. Then every node of degree two in $H$ is a $2-$node. 

\end{corollary}

%%%%%%%%%%%%%%%%%%%%%%%%%%%%%%%%%%%%%%%%%%%%%%%%%%%%%%%%%%%%%%%%%%%%%%%%%%%%%%%%%%%%%%%%%%%%%%%%%%%%%%%%%%%%%

\subsection{$2-$paths}
\label{subsec:coordinates-and-2-paths}
 Let $G$ be a \NE and $H\subseteq G$ a non-trivial $2-$edge connected component. A path $\pi = u_0-u_1-...-u_k$ in $H$ is called \emph{2-path} if $deg_H^-(u_i) = deg_H^+(u_i) = 1$ for every $0 < i < k$. The weights of the nodes enumerated in any $2-$path are denoted by using the same letter and subindex as the corresponding $2-$node, but in capital letters. For instance, the weight of the node $u_i$ is $U_i$. Notice that  whenever we consider a $2-$path $\pi = u_0-u_1-...-u_k$, either $u_i$ has bought exactly $(u_i,u_{i+1}) \in E(H)$ with $0 \leq i < k$, or $u_i$ has bought exactly $(u_i,u_{i-1})$ with $k \geq i > 0$. As a convention, we assume that in a $2-$path $\pi = u_0-u_1-...-u_k$ every $2-$node $u_i$ has bought exactly the link $(u_i,u_{i+1})$, with $0\leq i<k$.
 %(as can be seen in the picture below).  
 
%\begin{center}
%\begin{tikzpicture}
%\tikzset{vertex/.style = {shape=circle,draw,minimum size=1.5em,scale=0.75}}
%\tikzset{edge/.style = {->,> = latex'}}

% vertices

%\node[vertex] (u_0) at  (2,0) {$u_0$};
%\node[vertex] (u_1) at  (4,0) {$u_1$};
%\node[vertex] (u_2) at  (6,0) {$u_2$};
%\node (cdots) at  (8,0) {$\cdots$};
%\node[vertex] (u_k) at  (10,0) {$u_k$};

%edges

%\draw[edge] (u_0) to (u_1);
%\draw[edge] (u_1) to (u_2);
%\draw[edge] (u_2) to (cdots);
%\draw[edge] (cdots) to (u_k);

%\end{tikzpicture}
%\end{center}

 Notice that for every $2-$node $u_i$ of a $2-$path $\pi = u_0-u_1-...-u_k$ there exists a natural deviation that consists in swapping the link $(u_{i-1},u_{i})$ for the link $(u_{i-1},u_{i+1})$. Such deviation is called a \emph{$2-$swap} on $u_i$. 

 In the following we introduce the definition of a \emph{coordinate system} in a graph. Let $X$ be a subgraph of a $2-$edge connected component $H$ of $G$ and let $v_1,v_2 \in V(X)$. Let us assume that for any $x \in V(X)$ any shortest path between $x$ and $y \in V(G) \setminus V(X)$ passes through $v_1$ or $v_2$. In this situation we say that $\partial X = \left\{v_1,v_2 \right\}$ is the \emph{boundary} of $X$ and $\overline{X} = \cup_{x \in V(X) \setminus \partial X}T(x)$ is the \emph{interior} of $X$. Then we define $x_1(v),x_2(v)$ to be the distances in $G \setminus \overline{X}$ from $v \in V(G) \setminus V(X)$ to $v_1,v_2$, respectively. In this way, the application $(x_1,x_2) : V(G) \setminus V(X) \to \left(\mathbb{Z}\cup \left\{ \infty \right\} \right)^2$ is called a coordinate system and we use the notation $(\partial X)_1 = v_1$ and $(\partial X)_2=v_2$ to refer to such a coordinate system. Finally, if $a_1,a_2,b_1,b_2\in \mathbb{Z}$ then the expressions $\langle \underline{x_1+a_1},x_2+a_2\rangle$ and $[x_1+b_1,x_2+b_2]$ refer to the cardinality of the set of nodes $v \not \in \bar{X}$ such that $x_1(v)+a_1 \leq x_2(v)+a_2$ and $x_1(v)-x_2(v) = b_1 -b_2$, respectively. 
 
 Using the coordinate system given by the two endpoints of a $2-$path, we can measure the difference of costs of any node after applying a $2-$swap. 

\begin{lemma}	Let $G$ be a \NE and $H\subseteq G$ a non-trivial $2-$edge connected component. Let $\pi = u_0-u_1-...-u_k$ be a 2-path in $H$ such that $g(G) \geq 2k$. Let $\Delta C(u_i)$ be the cost difference associated to the $2-$swap on $u_{i+1}$ for $0 \leq i \leq k-2$. Then, using $(\partial \pi)_1 = u_k,(\partial \pi)_2 = u_0$ as a coordinate system we have that

$$\Delta C(u_i) = U_{i+1}-U_{i+2}-...-U_{k-1} - \langle \underline{x_1+k-i},x_2+i \rangle $$
\end{lemma}

\begin{proof}

 We have that: 

$$\Delta C(u_i) = U_{i+1}-U_{i+2}-...-U_{k-1} + \sum_{v \not \in \bar{\pi}} \Delta(v)$$

 Where $\Delta(v) = \min(x_1(v)+k-i-1,x_2(v)+i)- \min(x_1(v)+k-i,x_2(v)+i)$. Notice that $\Delta (v) \in \left\{-1,0\right\}$. More precisely, $\Delta (v) = -1$ iff $x_1(v)+k-i-1 < x_2(v)+i$. This is the same as saying that $x_1(v)+k-i \leq x_2(v)+i$. Now the conclusion easily follows.

\end{proof}

 Using similar arguments to the ones used in Proposition \ref{propo_diff} we can show an upper bound for the difference between the distance cost function evaluated in the two endpoints of any $2-$path of length 5 in any \NE.

\begin{proposition}
\label{proposition_2_path} Let $G$ be a \NE and $H\subseteq G$ a non-trivial $2-$edge connected component with $g(G) > 14$. Then, any $2-$path $\pi = u_0-u_1-....-u_5$ in $H$ satisfies 
$$D(u_5)-D(u_0) < 2n-(U_1+...+U_{4})$$
\end{proposition}

 \begin{proof}

 Indeed, since $H$ is $2-$edge connected there is at least one more neighbour of $u_5$ apart from $u_{4}$. We distinguish two cases depending whether $u_5$ has bought or not all the links to his neighbours except from the link $(u_{4},u_5)$. 

 First assume that there exists at least one neighbour distinct from $u_{4}$ that has bought a link to $u_5$. Let $v$ be such a node. Since there are not cycles with length strictly less than $12$ we must have $d_G(v,u_0) = 6$ and therefore $D(v) < D(u_0)+ n$ when considering the deviation that consists in deleting the link $(v,u_5)$ and buying a link to $u_0$. Now, as $v$ is adjacent to $u_5$ and $d_G(v,u_0) = 6$ as noticed before then $D(u_5) \leq D(v)+n-(U_1+...+U_{4})$. Therefore, combining the two inequalities we get $D(u_5)-D(u_0) < 2n-(U_1+...+U_{4})$. 

 Now suppose that except from the link $(u_{4},u_5)$, $u_5$ has bought all the links to the other neighbours. Let $x$ be a node minimizing the function $D(\cdot)$ over all nodes in $H$ and let $r$ be the subindex for which $u_5 \in A_{H,r}(x)$. If there is a neighbour $v \neq u_{4}$ of $u_5$ lying inside $A_{H,s}(x)$ with $s < r+1$ then when $u_5$ deletes the edge $(u_{5},v)$ and buys an edge to $x$ we get the inequality $D(u_5)<n+D(x)$ from where we deduce that $D(u_5) < n + D(x) \leq n + D(u_0)$, so we are done.  Otherwise, pick $v \in A_{H,r+1}(x)$ a neighbour of $u_5$ and let $v =w_0-w_1-...-w_l$ be a path obtained moving from $w_i$ to $w_{i+1}$ (with $(w_i,w_{i+1}) \in E(H)$ or $(w_{i+1},w_i)\in E(H)$), each time exactly one step farther away from $x$, ending in $w_l$, a node satisfying $n_{H,x}^+(w_l) = 0$. If $w_l$ has bought at least one link then $D(w_l) < n + D(x)$, considering the natural deviation in $w_l$.  Now consider two cases, depending whether $x = u_i$ with $1 \leq i \leq 4$ or not:

 First, suppose that $x \neq u_i$ for $1 \leq i \leq 4$. Then $r \geq 5$ so that the nodes $u_i$ are inside the sets $A_{H,r-(k-i)}(x)$ for $i=0,...,5$. This means that $D(u_5) \leq n-(U_1+...+U_{4})+D(w_l)$ when considering the deviation in $u_5$ that consists in moving the link $(u_5,v)$ to the link $(u_5,w_l)$. Combining this inequality with the previously obtained inequality $D(w_l) < n + D(x)$ we get the result.

 Secondly, consider that $x = u_i$ with $1 \leq i \leq 4$. We distinguish two cases:

 (i) First assume that $2+d_G(w_l,u_0)> d_G(u_1,u_5)$. Consider the deviation in $u_5$ that consists in swapping the link $(u_5,v)$ to the link $(u_5,w_l)$. In the new graph, the distance from $u_5$ to $u_i$ using the edge $(u_5,w_l)$ would be at least $i+d_G(w_l,u_0) +1$ whereas the distance using the edge $(u_4,u_5)$ is $d_G(u_5,u_i)\leq d_G(u_5,u_1)$, for any $i$ with $1 \leq i \leq 4$. But by the hypothesis, $d_G(u_5,u_1)<2+d_G(w_l,u_0)$. Therefore when $u_5$ deviates he does not use the edge $(u_5,w_l)$ to go to the nodes in $\cup_{i=1}^4T(u_i)$. As a consequence, imposing that $G$ is a \NE:

$$0 \leq n-(U_1+...+U_4)  +D(w_l)-D(u_5) $$

 But previously, we had seen that $D(w_l) < n + D(x)$. This last result implies that $0 < n -D(w_l)+ D(u_0)$. Therefore adding up these inequalities we get the conclusion:

$$0 < (n-(U_1+...+U_4)  +D(w_l)-D(u_5) )+ (n+D(u_0)-D(w_l)) =$$

$$ = 2n-(U_1+...+U_4)+D(u_0)-D(u_5)$$

 Which is equivalent to what we wanted to see.

 (ii) Otherwise, $d_G(w_l,u_0) \leq 2$. First, notice that $d_G(u_0,u_5)=5$ because $g(G) >14$. Thus, using the triangular inequality, $d_G(u_5,w_l) \leq d_G(u_5,u_0)+d_G(u_0,w_l) \leq 7$. Using this result and considering all the possible cases we reach a contradiction:

 (a) $d_G(u_0,w_l) = 0$ i.e, $w_l = u_0$. Since $u_0-u_1-...-u_5$ is a $2-$path then the path $u_0-u_1-...-u_5-w_0-w_1-...-w_l$ is a cycle. Moreover, such cycle has length at most $7+5 = 12$ which is a contradiction because $g(G)>14$.

 (b) $d_G(u_0,w_l) = 1$. Then, either $u_0 = w_{l-1}$, and then as before we can deduce that a cycle of length at most $5+6=11 < 14$ is obtained, or $u_0 \neq w_{l-1}$, and then again, as before, we can deduce that a cycle of length at most $1+5+7=13 < 14$ is obtained, a contradiction in both cases.

 (c) $ d_G(u_0,w_l) = 2$. Let $y$ be a node such that $u_0-y-w_l$ is a minimal length path between $u_0$ and $w_l$. If $u_0 = w_{l-2}$ then a cycle of length at most $5+5=10 < 14$ is obtained. Otherwise there are two subcases: either $y = w_{l-1}$ and then a cycle of length at most $1+5+6 = 12 < 14$ is obtained, or $y \neq w_{l-1}$ and then a cycle of length at most $5+2+7=14< g(G)$ is obtained. In all cases we reach a contradiction.

 Finally, the remaining case would be if $w_l$ has not bought any link, but then the same reasoning works applied to $w_{l-1}$ instead of $w_l$, so we are done.

\end{proof}

In contrast, we can show the following lower bound.

\begin{proposition} \label{prop_weight_edges} Let $G$ be a \NE and $H\subseteq G$ a non-trivial $2-$edge connected component with $g(G) > 14$. Then, any $2-$path $\pi = u_0-u_1-u_2-...-u_k$ in $H$ such that $k \geq 5$ satisfies $D(u_5)-D(u_0) \geq 2n -(U_1+...+U_4)$.
\end{proposition}
\begin{proof}
Let $\pi = u_0-u_1-... -u_5- ...-u_k$ be a 2-path in $H$. 
 Using $(\partial \pi)_1 = u_5,(\partial \pi)_2 = u_0$ as a coordinate system and applying Lemma 1 we get the following inequalities: $0 \leq U_4-\langle \underline{x_1+2},x_2+3 \rangle$,$0 \leq U_3-U_4- \langle \underline{x_1+3},x_2+2 \rangle$, $0 \leq  U_2-U_3-U_4-\langle \underline{x_1+4},x_2+1 \rangle$ and $0\leq U_1- U_2-U_3-U_4-\langle \underline{x_1+5},x_2 \rangle$.

 Furthermore, we have that:

$$D(u_5)-D(u_0) = 3U_1+U_2-U_3-3U_4+\sum_{v \not \in \bar{\pi} } \min(x_1(v),x_2(v)+5)-\min(x_1(v)+5,x_2(v))$$

 Therefore:

$$D(u_5)-D(u_0) \geq (U_1+U_2+U_3+U_4)+$$
$$+2\left(\langle \underline{x_1+5},x_2 \rangle+ \langle \underline{x_1+4},x_2+1 \rangle+\langle \underline{x_1+3},x_2+2 \rangle+\langle \underline{x_1+2},x_2+3 \rangle\right)+$$
$$+\sum_{v \not \in \bar{\pi} } \min(x_1(v),x_2(v)+5)-\min(x_1(v)+5,x_2(v))  = (U_1+U_2+U_3+U_4)+$$
$$+2([x+1,x]+[x,x])+4([x-1,x]+[x-2,x])+6([x-3,x]+[x-4,x])+8\left(\sum_{k \leq -5}[x+k,x] \right) +$$
$$+5\left(\sum_{k \geq 5}[x+k,x] \right)+\sum_{k = -4}^4 k[x+k,x]-5\left(\sum_{k \leq -5}[x+k,x] \right) \geq$$
$$\geq (U_1+U_2+U_3+U_4)+2(n-(U_1+U_2+U_3+U_4)) = 2n-(U_1+U_2+U_3+U_4).$$

\end{proof}

\begin{corollary}\label{prop_weight_edges_corol} Let $G$ be a \NE and $H\subseteq G$ a non-trivial $2-$edge connected component with $g(G) > 14$. If $\pi = u_0-u_1-u_2-...-u_k$ is a 2-path in $H$ then $k \leq 4$.

\end{corollary}

\subsection{An improvement towards the Tree Conjecture}
\label{subsec:an-imprpovement-towards-the-tree-conjecture}
 Let $G$ be a \NE for $\alpha > 4n$ and $H \subseteq G$ a non-trivial $2-$edge connected component. In such conditions, let $H_{\geq 3} = \left\{v \in V(H) \mid deg_H(v) \geq 3 \right\}$. Notice that we can consider the digraph $H'$ defined from $H$ setting $V(H') = H_{\geq 3}$ and $E(H')$ the set of edges $(u,v)$ with $u,v\in V(H')$ for which there is a maximal $2-$path $u =x_0-x_1-x_2-...-x_k = v$ (this is well-defined because of what we have shown in the previous sections). The weight of $e$ is then set to $k-1$ and we use the notation $w(e)$ to refer to the weight associated to the edge $e$. Finally, let $m = |V(H')|$. Now we are ready to prove that the average degree of $H$, which is noted as $\deg(H)$, is lower bounded by $2+\frac{1}{4}$. Then, combining this result jointly with the upper bound on the average degree obtained in \cite{Mihalaktree} we can show that every \NE $G$ for $\alpha > 17n$ is a tree. 

 The following Lemma is used to prove the next proposition.
\begin{lemma}
\label{lemma_mihalak_1}\cite{Mihalaktree} If $G$ is a \NE graph, $H \subseteq G$ a biconnected component of $G$, and  $u, v \in V(H)$ with
$d_G(u,v) \geq 3$ such that $u$ buys the edge to its adjacent vertex $x$ in a shortest $u-v$-path
and $v$ buys the edge to its adjacent vertex $y$ in that path, then $deg_H (x) \geq 3$ or
$deg_H (y) \geq 3$.

\end{lemma}

Even though in \cite{Mihalaktree} $H$ is a biconnected component, one can see that the result also holds when $H$ is a $2-$edge connected graph.
The next proposition is crucial to deduce the main result of this section.
\begin{proposition}
\label{proposition_weights} Let $G$ be a \NE and $H \subseteq G$ a non-trivial $2-$edge connected component. If $g(G)> 14$ then $deg(H) \geq 2+\frac{1}{4}$.

\end{proposition}

\begin{proof}

Let $H''$ be the subgraph obtained from $H'$ restricting to edges of strictly positive weight. If we see that $H''$ is a forest then this result together with Corollary \ref{prop_weight_edges_corol} would imply that $\sum_{e \in E(H')} w(e) \leq 3(m-1) < 3m$. With this assumption we could then conclude that:

$$deg(H) = \frac{\sum_{u \in V(H')}deg_H(u)+ 2 \sum_{e \in E(H')} w(e)}{\sum_{u \in V(H')}1+\sum_{e \in E(H')}w(e)}=$$

$$=2+\frac{\sum_{u\in V(H')}\left(deg_H(u)-2\right)}{m+\sum_{e \in E(H')}w(e)}  > 2+\frac{3m-2m}{m+3m} = 2+\frac{1}{4}$$

 Which is what we want to prove. Thus it is enough to show that $H''$ is a forest:

 Indeed, suppose for the sake of contradiction that $C=u_0-u_1-...-u_{k-1}-u_0$ is a minimal closed cycle in $H''$. First, notice that $C$ is directed or otherwise Lemma \ref{lemma_mihalak_1} could be used to reach a contradiction. Thus assume wlog that $e_i = (u_i,u_{i+1}) \in E(H')$ for each $i$ with $0\leq i \leq k-1$ are the edges that conform $C$, where the subindices are taken modulo $k$. 
Call $v_{i+1}$ the neighbours (from $G$) of $u_{i+1}$ lying in $e_i$ and let $w_{i+1}$ be the neighbours (from $G$) of $v_{i+1}$ in $e_i$ distinct than $u_{i+1}$, for every $i$ with $0 \leq i \leq k-1$, where the subindices are taken modulo $k$. The length of every maximal $2-$path associated to any edge $e \in E(H')$ is at most $4$ (again, by Corollary \ref{prop_weight_edges_corol}) and as a consequence we have that the inequalities:

$$2(d_G(w_i,u_i)+d_G(u_i,u_{i+1})) \leq 2(2+4) < 14 \leq g(G)$$

  Hold for every $i$ with $0 \leq i \leq k-1$, where the subindices are taken modulo $k$. This implies,  when considering the $2-$swaps on $v_i$, that $w_i$ gets further only from the nodes inside $T(v_i)$ (exactly one unit) and nearer from at least every node in $\cup_{x \in e_i}T(x)$ (exactly, again, one unit) for every $i$ with $0 \leq i \leq k-1$. Now, given an edge $e \in E(H')$ corresponding to a maximal $2-$path $x_0-x_1-...-x_k$ from $H$ we define $u \in e$ for  $u \in V(G)$ a node iff $u=x_i$ for some $i$ with $0 \leq i \leq k$. In this way, as a consequence of the last observations imposing that $G$ is a \NE we get the inequalities  $\sum_{x \in e_i}|T(x)| \leq V_i$ for each $i$ with $0 \leq i \leq k-1$. On the other hand, we have the obvious inequalities $\sum_{x \in e_i}|T(x)| > V_{i+1}$ because  $v_{i+1}$ belongs to $e_i$ for each $i$ with $0 \leq i \leq k-1$ where the subindices are taken modulo $k$. Combining together these inequalities we get: 

$$\sum_{i=0}^{k-1} V_i \geq \sum_{i=0}^{k-1} \sum_{x \in e_i}|T(x)| > \sum_{i=0}^{k-2}V_{i+1}+V_0 = \sum_{i=0}^{k-1} V_i$$ 

 Which is a contradiction. This implies that our first assumption was false and as a consequence we conclude that $H''$ is a forest, as we wanted to prove.

\end{proof}

Mamageishvili et al. in \cite{Mihalaktree} show a lower bound on the girth of any \NE graph $G$ as well as on the average degree of any $2-$edge connected component of $G$, both bounds in terms of $\alpha$ and $n$. 
\begin{theorem}
\label{propo_mihalak} \cite{Mihalaktree} Let $G$ be a \NE. Then $g(G) \geq 2\frac{\alpha}{n}+2$.
\end{theorem}
\begin{lemma} \cite{Mihalaktree} 
\label{mihalak-upperbound} Let $G$ be a \NE for $\alpha > n$ and $H \subseteq G$ a biconnected component of $G$. Then, $deg(H) \leq 2 + \frac{4n}{\alpha-n}$.
\end{lemma}

Even though this lemma  is stated for biconnected components in \cite{Mihalaktree}, it is not hard to see that the  proof also works for $2-$edge connected graphs.
 Hence, combining the previous bounds jointly with Proposition \ref{proposition_weights}, we can enlarge the interval of $\alpha$ for which every \NE is a tree.
\begin{theorem} \label{thm_tree}
For $\alpha > 17n$ every \NE is a tree.
\end{theorem}

\subsection{A further improvement on the Price of Anarchy}
 \label{subsec:a-further-improvement-on-the-price-of-anarchy}
 
  In this section we show  that the $PoA$ is constant even for $\alpha > 9n$. In order to do so,  recall that in order to bound the $PoA$ it is enough to bound the diameter of any \NE. 
 \begin{lemma} \label{lemma_diam}\cite{Demaineetal:07} Let $G$ be a \NE for $\alpha \geq 2$. Then, the $PoA$ is upper bounded by $diam(G)+1$.
 \end{lemma}
  Since we are mainly working with a non-trivial $2-$edge connected component $H$ of a \NE graph $G$, it seems natural to find a relation between the diameter of $G$ and the diameter of $H$.
 We can show that in any non-trivial $2-$edge connected component $H$ of a \NE $G$,  the depth of any connected component $T(u)$ for $u \in V(H)$ is upper bounded by a constant. This result allows us to prove the following relation between $diam(G)$ and $diam(H)$: 

\begin{proposition} 
\label{diamH-upperbound}

Let $G$ be a \NE for $\alpha > 4n$ and $H \subseteq G$ a nontrivial $2-$edge connected component of $G$. Then, $diam(G) \leq diam(H)+206$.

\end{proposition}

\begin{proof}

 For $\alpha > 17n$ every \NE is a tree so in this case there do not exist any nontrivial  $2-$edge connected component. Therefore it is enough to show that $diam(G) \leq diam(H)+206$ when $\alpha \leq 17n$.

 Indeed, let $u',v'$ be nodes such that $d_G(u',v') = diam(G)$. Assume that $u,v\in V(H)$ are the nodes such that $u' \in T(u),v'\in T(v)$. Let $l_u = d_G(u,u')$ and $l_v = d_{G}(v,v')$. Since $G$ is a \NE, if $u'$ buys a link to $u$ then it holds that: 

$$0 \leq \alpha -(l_u-1)(n-U) \leq  17n-(l_u-1)(n-U)\Rightarrow l_u \leq 1 + \frac{17n}{n-U}$$

 Similarly, $l_v \leq 1+ \frac{17n}{n-V}$. Next, let $z,t$ be nodes at maximum distance from $u,v$ respectively. Since $\alpha > 2n$ by Proposition \ref{propo_mihalak} the girth of $G$ is greater than or equal $2(2+1) = 6$ so that $d_G(u,z),d_G(v,t) \geq 3$. Also, notice that $|d_G(z,x)-d_G(u,x)| \leq d_G(u,z)$ by the triangular inequality. Using this together with Corollary \ref{corol_diff_cost} we obtain: 

$$2n> D(z)-D(u) = \sum_{x \in V(G)} \left( d_G(z,x)-d_G(x,u)  \right) = $$

$$= \sum_{x \in T(u)} \left( d_G(z,x)-d_G(x,u)  \right) + \sum_{x \not \in T(u)}\left(d_G(z,x)-d_G(x,u) \right)  \geq $$

$$ \geq \sum_{x \in T(u)}d_G(u,z) + \sum_{x \not \in T(u)}(-d_G(u,z)) = d_G(u,z)U-d_G(u,z)(n-U) =$$

$$= d_G(u,z)(2U-n) $$

 Hence, either $U \leq n/2$ and then $l_u \leq 1 + \frac{17}{1/2} = 35$ or if $U > n/2$ then the previous inequality implies that $U < \frac{2n/3+n}{2} = \frac{5}{6}n$ so that $l_u \leq 1+\frac{17}{1/6} = 103$. Therefore: 

$$diam(G) \leq l_u+diam(H)+l_v \leq diam(H)+206$$

 As we wanted to see.

\end{proof}

In the following, we are going to improve the lower bound for $deg(H)$. The basic idea is to analyze the structure that form the edges from $H''$ in a bit more of detail than we did in Proposition \ref{proposition_weights}. In there, we exploited the fact that there do not exist cycles of edges from $H''$, thus deducing that $H''$ is a forest. This approach could be regarded as a kind of linear exploration, in the sense that we only deviated in the direction that the $2-$nodes forming the edges from $H''$ define, which is unique. Recall that for any $u \in V(H'')$, $deg_{H''}^-(u) \leq 1$ and $deg_{H'}(u) \geq 3$. Then, we can consider the following deviation: $u \in V(H'')$ deletes two links (from $H$) and buys a link to a node close to it. If the variation of the sum of the distances to the other nodes is small enough, then this deviation could represent an advantage to $u$. This is exactly what we are going to use in the following two lemmas.
% in order to build the basic results that will be used later in Proposition \ref{propFinal}. 

First, we need to extend the definitions we made in the previous section about coordinate systems. In this new scenario we are dealing with a subgraph $X$ and three nodes $v_1,v_2,v_3 \in X$ having the exact same properties as in the case of cardinality two. The same definitions work except that now the boundary of $X$, which is called $\partial X$, has three elements, $v_1,v_2,v_3$. Then defining analogously $x_1,x_2,x_3$, we obtain a coordinate system of cardinality three that is noted as $(\partial X)_1 =v_1,(\partial X)_2 = v_2,(\partial X)_3 = v_3$.   Let $a_1,a_2,a_3\in \mathbb{Z}$, then: $\langle \underline{x_1+a_1},x_2+a_2,x_3+a_3\rangle$  refer to the cardinality of the set of nodes $v \not \in \bar{X}$ such that $x_1(v)+a_1 \leq x_2(v)+a_2,x_3(v)+a_3$. Similarly, $\langle x_1+a_1, \underline{x_2+a_2},x_3+a_3\rangle$ refer to the cardinality of the set of nodes $v \not \in \bar{X}$ such that $x_2(v)+a_2 \leq x_1(v)+a_1,x_3(v)+a_3$. Finally, $\langle \underline{x_1+a_1}, \underline{x_2+a_2}, x_3+a_3 \rangle$ refer to the cardinality of the set of nodes $v \not \in \bar{X}$ such that $\min(x_1(v)+a_1,x_2(v)+a_2) \leq x_3(v)+a_3$.

%%%%%%%%%%%%%%%%%%%%%%%%%%%%
%%%%%%
%%%%%%
%%%%%%%
%%%%%%%
%%%%%%%%%%%%%%%%%%%%%%%%%%%%%

\begin{lemma} Let $G$ be a \NE for $\alpha > 4n$ and $H \subseteq G$ a non-trivial $2-$edge connected component. Assume that $g(G) \geq 16$ and $diam(H) \geq 62$. Let $\pi = u_0-u_1-u_2-...-u_k$ be a path in $H$ having at least three $2-$nodes, with $k \leq 7$. Then there cannot be more than one $2-$path $\pi' = u_k-...-u_{k+l}$ with $l \geq 2$. 
\label{lemma_pes_0}
\end{lemma}

\begin{proof}
 Assume the contrary and we see that a contradiction is reached. Indeed assume that $\pi_1 = u_k-v_1-w_1$, $\pi_2 = u_k-v_2-w_2$ are two $2-$paths of length two adjacent to $u_k$ and disjoint with $\pi$. Let $u_{i_1},u_{i_2},u_{i_3}$ be three $2-$nodes from $\pi$ and let $(\partial \phi)_1 = w_1, (\partial \phi)_2 = w_2$ and $(\partial \phi)_3 = u_k$ be a coordinate system, where $\phi = \pi_1 \cup \pi_2$. Consider the $2-$swaps on $u_{i_1},u_{i_2},u_{i_3}$. Since $k \leq 7$ and $g(G) \geq 16=2\cdot (7+1)$, when considering the $2-$swap on $u_{i_1}$, the node $u_{i_1-1}$, sees every node in $T(v_1),T(v_2),T(u_{i_2})$ and $T(u_{i_3})$ one unit closer than before deviating. Also, it is clear that in such deviation every node in $T(u_{i_1})$ gets one unit further from $u_{i_1-1}$ and that there are no more nodes in $G$ having this property. Thus, imposing that $G$ is a \NE we get the inequality:

$$U_{i_1} \geq U_{i_2}+U_{i_3}+V_1+V_2$$

 In a similar way, if we consider the $2-$swaps on $u_{i_2}$ and $u_{i_3}$ we get the inequalities:

$$U_{i_2} \geq U_{i_3}+V_1+V_2$$

 And

$$U_{i_3} \geq V_1+V_2$$

 Respectively. 

 Now consider the $2$-swap on $v_1$. For any node $v\in V(G)$ the distance change to $u_k$ belongs to the set $\left\{-1,0,1\right\}$. Clearly, $u_k$ gets one unit further from every node inside $T(v_1)$ when deviating and from the remaining nodes the distance change is either $-1$ or $0$. More precisely, from this remaining set of nodes, $u_k$ gets one unit closer exactly to the nodes $v \not \in \bar{\phi}$ satisfying $x_1(v)+1 < x_2(v)+2,x_3(v)$, which is the same as saying $x_1(v)+2 \leq x_2(v)+2,x_3(v)$. Therefore, imposing that $G$ is a \NE we get the following inequality:

$$V_1 \geq \langle \underline{2+x_1},2+x_2, x_3 \rangle $$ 

 Likewise, considering the same reasoning in $v_2$: 

$$V_2 \geq \langle 2+x_1, \underline{2+x_2}, x_3 \rangle $$
%$$T_2^1 \geq |\left\{v \mid x_2(v)+2 \leq \min(2+x_1(v),x(v)) \right\}|$$

 By Proposition \ref{propo_cicles1} there exist two minimal cycles $c_1,c_2$ passing through $\pi_1, \pi_2$, respectively. Also, by Proposition \ref{propo_cicles2} neither of $c_1,c_2$ contains simultaneously the two $2-$paths $\pi_1,\pi_2$. This implies that when we delete the links $(u_k,v_1),(u_k,v_2)$ we can use $c_i$ to go from $u_k$ to $v_i$ for $i=1,2$. 
Consider the deviation that consists in deleting the edges $(u_k,v_1),(u_k,v_2)$ and adding a link to $u_{i_1}$ and call $\Delta C_1$ the corresponding cost difference. 

 For this deviation, notice that any node $v \not \in \bar{\phi}$ gets further from $u_k$ in the deviated network iff $x_3(v) \leq \min(x_1(v)+2,x_2(v)+2)$ does not hold, i.e, iff $\min(x_1(v)+2,x_2(v)+2) < x_3(v)$. Also, for such set of nodes, the corresponding distance change in the deviated network is of at most $(l(c_i)-2)-2 \leq (2f+1-2)-2 < 2f$, where $f = diam_H(u_k)$, because as said before, we can use the cycle $c_i$ and we know that $l(c_i) \leq 2f+1$ because $c_i$ is minimal. 

 Using similar arguments, it can be shown that the remaining nodes $v \in \bar{\phi} = T(v_1) \cup T(v_2)$ gets further from $u_k$ in the deviated network, too, and that the corresponding distance change is also upper bounded by $2f$. 

 Therefore, imposing that $G$ is a \NE we get:

$$ \Delta C_1 < -\alpha +2f\left(V_1+V_2+\langle \underline{2+x_1}, \underline{2+x_2},x_3\rangle \right)$$ %nova versio
%$$ \Delta c_1 \leq -\alpha -(k-2)U_{i_1} + (2g-2)(T^1_1+T^2_1+|A_{z_k}(t^1_1,t^2_1)|)$$ %nova versio

 On the other hand, let $z$ be a node at the maximum distance from $u_k$, i.e a node verifying $d_H(u_k,z)=f$, and consider the deviation that consists in adding a link from $z$ to $u_k$. Call $\Delta C_2$ the cost difference associated to such deviation. Notice that the distance change (in absolute value) associated to each node in $T(u_{i_1})$ is at least $d_H(z,u_{i_1})-(1+d_H(u_k,u_{i_1})) \geq f-2d_{H}(u_k,u_{i_1})-1 \geq f-13 $, using the triangular inequality together with $k \leq 7$. Moreover, the same upper bound works if we consider the nodes in $T(u_{i_2})$ and $T(u_{i_3})$. Therefore, imposing that $G$ is a \NE we get:

$$ \Delta C_2 < \alpha - (f-13)\left(U_{i_1}+ U_{i_2}+U_{i_3}\right)$$

 Thus, adding these two inequalities and combining the resulting inequality with the previous ones we get:
$$ \Delta C_1+\Delta C_2 < - (f-13)\left(U_{i_1}+ U_{i_2}+U_{i_3}\right)+2f \left( V_1+V_2+\langle \underline{2+x_1}, \underline{2+x_2},x_3\rangle \right) \leq $$
$$ \leq  -2(f-13) \left(U_{i_2}+U_{i_3}\right)+(f+13) \left(V_1+V_2\right)+2f \langle \underline{2+x_1}, \underline{2+x_2},x_3\rangle \leq $$
$$ \leq -4(f-13) U_{i_3}-(f-39)\left( V_1+V_2\right)+2f\langle \underline{2+x_1}, \underline{2+x_2},x_3\rangle \leq $$
$$\leq \left(-5f+91\right)(V_1+V_2)+2f\langle \underline{2+x_1}, \underline{2+x_2},x_3\rangle $$
 But

$$V_1+V_2 \geq \langle \underline{2+x_1}, 2+x_2,x_3\rangle+\langle 2+x_1, \underline{2+x_2},x_3\rangle \geq  \langle \underline{2+x_1}, \underline{2+x_2},x_3\rangle$$ 

 Therefore:

$$\Delta C_1+\Delta C_2 < \left(-3f+91\right)\langle \underline{2+x_1}, \underline{2+x_2},x_3\rangle \leq  0$$ 

 Because by assumption $diam(H) \geq 62$ so that $f \geq 31$ by the triangular inequality.

\end{proof}

\begin{lemma} Let $G$ be a \NE for $\alpha > 4n$ and $H$ a nontrivial $2-$edge connected component. Assume that $g(G) \geq 12$ and that $diam(H) \geq 126$.  Let $\pi = u_0-u_1-u_2-...-u_k$ be path in $H$ having at least two $2-$nodes, with $k \leq 4$. Then there cannot be more than one $2-$path $\pi' = u_k-...--u_{k+l}$ with $l > 2$. 

\label{lemma_pes_1}
\end{lemma}

\begin{proof} 

 Assume the contrary and we see that a contradiction is reached. Indeed assume that $\pi_1 = u_k-v_{1}-w_1-t_1$, $\pi_2 =u_k-v_2-w_2-t_2$ are two $2-$paths of length two adjacent to $u_k$ and disjoint with $\pi$. Let $u_{i_1},u_{i_2}$ be two $2-$nodes from $\pi$ and let $(\partial \phi)_1 = t_1, (\partial \phi)_2 = t_2$ and $(\partial \phi)_3 = u_k$ be a coordinate system, where $\phi = \pi_1 \cup \pi_2$. Consider the $2-$swaps on $u_{i_1},u_{i_2}$. Since $g(G) \geq 12=2\cdot (4+2)$, when considering the $2-$swap on $u_{i_1}$, the node $u_{i_1-1}$, sees every node in $T(v_1),T(v_2),T(w_1),T(w_2)$ and $T(u_{i_2})$ one unit closer than before deviating. Also, it is clear that in such deviation every node  in $T(u_{i_1})$ gets one unit further from $u_{i_1-1}$ and that there are no more nodes in $G$ having this property. Thus, imposing that $G$ is a \NE we get the inequality: 

$$U_{i_1} \geq U_{i_2}+V_1+V_2+W_1+W_2$$

 In a similar way, if we consider the $2-$swap on $u_{i_2}$ we get the inequality:

$$U_{i_2} \geq V_1+V_2+W_1+W_2$$

 Now consider the $2$-swap on $v_1$. For any node $v\in V(G)$ the distance change to $u_k$ belongs to the set $\left\{-1,0,1\right\}$. Clearly, $u_k$ gets one unit further from every node inside $T(v_{1})$ when deviating and from the remaining nodes the distance change is either $-1$ or $0$. More precisely, from this remaining set of nodes, $u_k$ gets one unit closer exactly to the nodes $v \not \in \bar{\phi}$ satisfying $x_1(v)+2 < x_2(v)+3,x_3(v)$, which is the same as saying $x_1(v)+3 \leq x_2(v)+3,x_3(v)$. Therefore, imposing that $G$ is a \NE we get the following inequality:

$$V_1 \geq W_1+ \langle \underline{3+x_1},3+x_2, x_3 \rangle $$ 

 Likewise, considering the same reasoning in $v_2$: 

$$V_2 \geq W_2 + \langle 3+x_1, \underline{3+x_2}, x_3 \rangle $$

 In a similar way, considering the $2-$swaps on $w_1$ and $w_2$ and imposing that $G$ is a \NE we get the inequalities: 

$$W_1 \geq \langle \underline{2+x_1},4+x_2,1+x_3 \rangle$$
$$W_2 \geq \langle 4+x_1, \underline{2+x_2},1+x_3 \rangle$$

%%%%%%%%%%%%%%%%%%%%%%%%%%%%%%%%%%%%%%%%%%%%%%%%%%%%%%%%%%%%%%%%%%%%%%%%%%%%%%%%%

 Now, by Proposition \ref{propo_cicles1} there exist two minimal cycles $c_1,c_2$ passing through $\pi_1, \pi_2$, respectively. Also, by Proposition \ref{propo_cicles2} neither of $c_1,c_2$ contains simultaneously the two $2-$paths $\pi_1,\pi_2$. This means that when we delete the edges $(u_k,v_1),(u_k,v_2)$ we can use $c_i$ to go from $u_k$ to $v_i$ and $w_i$ for $i=1,2$. 
Consider the deviation that consists in deleting the edges $(u_k,v_1),(u_k,v_2)$ and adding a link to $u_{i_1}$ and call $\Delta C_1$ the corresponding cost difference. 

 For this deviation, notice that any node $v \not \in \bar{\phi}$ gets further from $u_k$ in the deviated network iff $x_3(v) \leq \min(x_1(v)+3,x_2(v)+3)$ does not hold, i.e, iff $\min(x_1(v)+3,x_2(v)+3) < x_3(v)$. Also, for such set of nodes, the corresponding distance change in the deviated network is of at most $(l(c_i)-3)-3 \leq (2f+1-3)-3 < 2f$, where $f = diam_H(u_k)$, because as said before, we can use the cycle $c_i$ and we know that $l(c_i) \leq 2f+1$ because $c_i$ is minimal. 

 Using similar arguments, it can be shown that the remaining nodes $v \in \bar{\phi} = T(v_1) \cup T(v_2) \cup T(w_1) \cup T(w_2)$ gets further from $u_k$ in the deviated network, too, and that the corresponding distance change is also upper bounded by $2f$. 

 Therefore, imposing that $G$ is a \NE we get:

$$ \Delta C_1 < -\alpha +2f\left(V_1+V_2+W_1+W_2+\langle \underline{3+x_1}, \underline{3+x_2},x_3\rangle \right)$$ %nova versio
%$$ \Delta c_1 \leq -\alpha -(k-2)U_{i_1} + (2g-2)(T^1_1+T^2_1+|A_{z_k}(t^1_1,t^2_1)|)$$ %nova versio

 On the other hand, let $z$ be a node at the maximum distance from $u_k$, i.e a node verifying $d_H(u_k,z)=f$, and consider the deviation that consists in adding a link from $z$ to $u_k$. Call $\Delta C_2$ the cost difference associated to such deviation. Notice that the distance change (in absolute value) associated to each node in $T(u_{i_1})$ is at least $d_H(z,u_{i_1})-(1+d_H(u_k,u_{i_1})) \geq f-2d_{H}(u_k,u_{i_1})-1 \geq f-7 $, using the triangular inequality together with $k \leq 4$. Moreover, the same upper bound works if we consider the nodes in $T(u_{i_2})$. Therefore, imposing that $G$ is a \NE we get:

$$ \Delta C_2 < \alpha - (f-7)\left(U_{i_1}+ U_{i_2}\right)$$
%%%%%%%%%%%%%%%%%%%%%%%%%%%%%%%%%%%%%%%%%%%%%%%%%%%%%%%%%%%%%%%%%%%%%%%%%%%%%%%%%%%%%%%%%%%%%%%

 Thus, adding these two inequalities and combining the resulting inequality with the previous ones we get:

$$ \Delta C_1+\Delta C_2 < - (f-7)\left(U_{i_1}+ U_{i_2}\right)+2f \left( V_1+W_1+V_2+W_2+\langle \underline{3+x_1}, \underline{3+x_2},x_3\rangle \right) \leq $$%nova versio
$$ \leq  -2(f-7)U_{i_2}+(f+7)\left(V_1+W_1+V_2+W_2\right)+2f \langle \underline{3+x_1}, \underline{3+x_2},x_3\rangle \leq $$
$$ \leq -(f-21)(V_1+W_1+V_2+W_2)+2f \langle \underline{3+x_1}, \underline{3+x_2},x_3\rangle \leq $$
$$  \leq -(3(f-21)-2f) \langle \underline{3+x_1}, \underline{3+x_2},x_3 \rangle \leq 0$$

 Where we have used that 

$$\langle \underline{2+x_1},4+x_2, x_3+1 \rangle+\langle 4+x_1, \underline{2+x_2}, x_3+1 \rangle \geq \langle \underline{3+x_1}, \underline{3+x_2},x_3\rangle $$ 

 And that $diam(H) \geq 126$ so that $f \geq 63$ by the triangular inequality.

\end{proof}

 Starting at an arbitrary edge from $E(H')$ of positive weight, we can construct a walk of adjacent edges through $E(H')$ in such a way that the average of the weights of the edges from the walk is small enough. The basic idea is that if we have already built a path $\pi$ of edges from $E(H')$ and we are currently standing on the edge $e \in E(H')$, we can apply Lemma \ref{lemma_pes_0} and Lemma \ref{lemma_pes_1} to $\pi$ to deduce that, among all the unvisited edges adjacent to $e$, we can choose at least one edge $e'$ of weight $0$ or $1$ such that  when adding $e'$ to $\pi$ the average of the weights of the edges from $\pi$ is reduced. For an appropriate girth, we can build walks of this kind guaranteeing that the sets of visited edges in each walk are mutually disjoint. In this way, we can diminish the average weight of the edges of positive weight from $E(H')$ thus getting a higher lower bound for $deg(H)$.

% ********** These two technical lemmas help us to show a lower bound on the average degree of a non-trivial biconnected component $H$ of a \NE $G$, for appropriate conditions on the girth of $G$, the diameter of $H$ and the range for the parameter $\alpha$. 

\begin{proposition} 
\label{propFinal}
Let $G$ be a \NE for $\alpha > 4n$ and $H\subseteq G$ a $2-$edge connected component of $G$. If $g(G) \geq 20$ and $diam(H) \geq 126$ then $deg(H) \geq 2+\frac{1}{2}$.
\end{proposition}

\begin{proof}

 Looking at the proof of Proposition \ref{proposition_weights} it is enough to show that $\sum_{e \in E(H')}w(e) \leq \frac{2}{3}|E(H')|$. To this purpose, we shall notice that the edges in $E(H')$ can be grouped into disjoint subsets in such a way that the average of the weights of the edges in every subset is at most $2/3$.

 Equivalently, we show that we can associate to each edge in $e\in E(H')$ of strictly positive weight a subset $\phi(e)$ of edges of weight $0$ and $1$ such that the average of the weights of the edges of $\phi(e)$ together with the weight of $e$ is at most $2/3$. Also, the association is such that there is no associated edge of weight $0$ or $1$ belonging simultaneously to two subsets $\phi(e_1)$ and $\phi(e_2)$, for two distinct edges $e_1,e_2$ of strictly positive weight.

 Before proving the result we introduce some notation. Given two edges $f_1,f_2 \in E(H')$ and a node $z\in V(H')$ we say that $f_2$ is adjacent to $f_1$ in $z$ if $f_1,f_2$ share exactly the endpoint $z$. Moreover, let $f_1,f_2,...,f_k \in E(H')$ be a sequence of edges. We say that $f_1,f_2,...,f_k$ is a non-degenerate sequence of consecutive edges if for every $i$ with $1 \leq i \leq k-1$, $f_i,f_{i+1}$ share exactly one endpoint, call it $x_i$, and $x_i \neq x_j$ for every $ i,j$ with $i \neq j$. 

 Now, consider $e =(u,v) \in E(H')$ an edge of strictly positive weight. We define the association $\phi(e)$ in the following way:

 (i) If $w(e)=3$. Start with the path induced by $e$. By Lemma \ref{lemma_pes_0} there exists an edge $e_1 \in E(H')$ of weight zero such that $e_1$ is adjacent to $e$ in $v$. Now consider the path induced by $e$ together with $e_1$. Analogously, Lemma \ref{lemma_pes_0} tells us that there exists an edge $e_2 \in E(H')$ of weight zero such that $e,e_1,e_2$ is a non-degenerate sequence of consecutive edges. We can repeat this argument at least two more times to find two more edges of weight zero, call them $e_3,e_4$, such that $e,e_1,e_2,e_3,e_4$ is a non-degenerate sequence of consecutive edges.  Then, letting $\phi(e) = \left\{ e_1,e_2,e_3,e_4\right\}$ the desired requirements are fulfilled since $3/5 < 2/3$. 

 (ii) If $w(e) = 2$. By Lemma \ref{lemma_pes_1} then we find an edge $e_1 \in E(H')$ of weight at most one adjacent to $e$ in $v$. We consider two subcases:  

 (a) If $w(e_1) = 0$. Then again, we can apply Lemma \ref{lemma_pes_1} to the path induced by $e,e_1$ and find an edge $e_2 \in E(H')$ of weight at most $1$ such that $e,e_1,e_2$ is a non-degenerate sequence of consecutive edges. If $w(e_2) = 0$ then letting $\phi(e) = \left\{e_1,e_2 \right\}$ the requirements are fulfilled. Otherwise, if $w(e_2) = 1$ then we can apply Lemma \ref{lemma_pes_0} two times and find two consecutive edges $e_3,e_4 \in E(H')$ of weight $0$ such that $e,e_1,e_2,e_3,e_4$ is a non-degenerate sequence of consecutive edges. Thus letting $\phi(e) = \left\{ e_1,e_2,e_3,e_4 \right\}$ then  the desired conditions are satisfied.

 (b) If $w(e_1) = 1$. Then we can apply Lemma \ref{lemma_pes_0} three times to find edges $e_2,e_3,e_4 \in E(H')$ of weight $0$ defining together with $e,e_1$, again, a non-degenerate sequence of consecutive edges. Likewise, $\phi(e) = \left\{ e_1,e_2,e_3,e_4\right\}$ satisfies the desired conditions.

 (iii) If $w(e)=1$.  We distinguish four subcases: 

 (a) There exists an edge $e_1 \in E(H')$ of weight $0$ adjacent to $e$ in $v$. Then letting $\phi(e) = \left\{e_1 \right\}$ we get that the corresponding average weight is $1/2 < 2/3$ so we are done.

 (b) If there is no edge adjacent to $e$ in $v$ of weight zero but there exists at least an edge $e_1\in E(H')$ of weight $1$. Then we can apply Lemma \ref{lemma_pes_1} to the path induced by $e,e_1$ and find at least one edge $e_2 \in E(H')$ of weight at most $1$ such that $e,e_1,e_2$ is a non-degenerate sequence of consecutive edges. If $w(e_2) = 0$ then $\phi(e) = \left\{e_1,e_2 \right\}$ satisfies the requirements. Otherwise, if $w(e_2) =1$ then we can apply two times Lemma \ref{lemma_pes_0} to obtain edges $e_3,e_4 \in E(H')$ of weight zero such that $e,e_1,e_2,e_3,e_4$ is a non-degenerate sequence of consecutive edges so that $\phi(e) = \left\{e_1,e_2,e_3,e_4\right\}$ satisfies the desired conditions.

 (c) If there is no edge adjacent to $e$ in $v$ of weight at most one but there exists at least an edge $e_1 \in E(H')$ of weight $2$. Then, again, we can apply Lemma \ref{lemma_pes_0} three times to obtain edges $e_2,e_3,e_4 \in E(H')$ of weight zero such that $e,e_1,e_2,e_3,e_4$ is a non-degenerate sequence of consecutive edges. Setting $\phi(e) = \left\{e_1,e_2,e_3,e_4 \right\}$ the required conditions are fulfilled. 

 (d) If every edge adjacent to $e$ in $v$ distinct than $e$ has weight $3$. Let $e_1 \in E(H')$ any edge adjacent to $e$ in $v$ distinct than $e$. Then we can apply four times Lemma \ref{lemma_pes_0} to obtain edges $e_2,e_3,e_4,e_5 \in E(H')$ of weight zero such that $e,e_1,e_2,e_3,e_4,e_5$ is a non-degenerate sequence of consecutive edges. Finally, if we let $\phi(e) = \left\{ e_1,e_2,e_3,e_4,e_5\right\}$ then the average of weights of $e$ together with the weights in $\phi(e)$ is exactly $2/3$, as we wanted to see.

 Now assume that $\phi$ has been defined for every edge of weight $1,2,3$. Next we show that for any two distinct edges $e_1=(u_1,v_1),e_2=(u_2,v_2) \in E(H')$ of weight $1,2$ or $3$ there is no edge $e$ belonging simultaneously to $\phi(e_1)$ and $\phi(e_2)$.

 For the sake of contradiction, assume the contrary, that is, that there is an edge $e=(u,v) \in E(H')$ with $e \in \phi(e_1) \cap \phi(e_2)$. By construction of the association $\phi$, for any edge $e' \in E(H')$, the elements of $\phi(e')$ together with $e'$ define a non-degenerate path. Let $\pi_1,\pi_2$ be these corresponding paths for $\phi(e_1),\phi(e_2)$, respectively. Now define $x_1,x_2$ to be the nearest endpoint of $e$  we found when we follow $\pi_1,\pi_2$ starting from $e_1,e_2$,  respectively. Also, let $z_1,z_2$ be the nearest $2-$nodes to $x_1,x_2$, that we find when following $\pi_1,\pi_2$ starting at $e_1,e_2$, respectively (such $2-$nodes must exist because $e_1,e_2$ have strictly positive weight). Assume that $(z_1,w_1),(t_1,z_1) \in E(H')$ and $(z_2,w_2),(t_2,z_2) \in E(H')$ are the edges that define the $2-$nodes $z_1,z_2$, respectively. As it can be seen from the construction of the association $\phi$ we need no more than $5$ edges (from $G$) to go from $t_1$ to $x_1$ through $\pi_1$ and from $t_2$ to $x_2$ through $\pi_2$, respectively.  Now suppose wlog that $Z_1 \leq Z_2$. Finally, notice that if $x_1 = x_2$ then $d_G(t_1,z_2) \leq 5+4 = 9$ whereas if $x_1 \neq x_2$ then $d_G(t_1,z_2) \leq 5+1+4 = 10$. Thus in both cases when $t_1$ swaps the link $(t_1,z_1)$ for the link $(t_1,w_1)$ he sees node $z_2$ exactly one unit closer than before because $g(G) \geq 2 \cdot 10$ by hypothesis. In this way the cost difference associated to such deviation is strictly less than $Z_1 - Z_2 \leq 0$ units, a contradiction since $G$ is a \NE. 

\end{proof}

 As a consequence of all previous results we get the following theorem:

\begin{theorem} \label{thm_9n}
For $\alpha > 9n$ the $PoA$ is constant.
\end{theorem}

\begin{proof} 

It is well known that the $PoA$ for trees is at most $5$.  Thus, let $G$ be any \NE graph for $\alpha > 9n$ having at least one non-trivial $2-$edge connected component $H$. Since $\alpha > 9n$ then $g(G) \geq 20$ by Theorem \ref{propo_mihalak}. If $diam(H) \geq 126$ since $g(G) \geq 20$ then by Proposition \ref{propFinal} we have that $deg(H) \geq 2 +1/2$. However, by Lemma \ref{mihalak-upperbound}, $deg(H) \leq 2 + 4n/(\alpha-n )< 2+1/2$ which is a contradiction. Hence, $diam(H) < 126$. Then, by Proposition \ref{diamH-upperbound} and Lemma \ref{mihalak-upperbound} we have that $PoA \leq 332$.

\end{proof}

\section{The lower range}
\label{sec:so-and-ne}

 In this section we focus our attention to the range $\alpha < n/C$ with $C > 4$, which we call the lower range. More precisely, we show that there is a connection between the equilibria for this range and the so called distance-uniform and distance-almost-uniform graphs introduced by Alon et al. in \cite{AlonDHKL14}.

%\subsection{Distance-almost-uniform graphs}

 A graph $G$ is \emph{$(k,\epsilon)-$distance-almost-uniform} if there exists an $r$ such that $\max_{i=0}^{k-1} |A_{r+i}(u)| \geq n(1-\epsilon)$ for all $u\in V(G)$, where $A_s(v)$ is the set of nodes $w$ at distance exactly $s$ from $v$. In this way, $\epsilon-$distance-uniform and $\epsilon-$distance-almost-uniform as introduced in \cite{AlonDHKL14} correspond to our $(1,\epsilon)-$distance-almost-uniform and $(2,\epsilon)-$distance-almost-uniform definitions, respectively. Now, a collection of graphs $\mathcal{F}$ is $k$-distance-almost-uniform if there exists a constant $\epsilon < 1$ for which every $F \in \mathcal{F}$ is $(k,\epsilon)-$distance-almost-uniform.

\begin{proposition}
\label{proposition-distance-almost-uniform-1}

Let $C > 4$ be a positive constant. Then every \NE $G$ for $\alpha < n/C$ is $(5,\epsilon)-$almost-distance-uniform for $\epsilon = \frac{4}{5}(1+1/C)$. 

\end{proposition}

\begin{proof}

 First, fix $u \in V(G)$. Consider for $w \neq u$ the two deviations that consist in adding a link from $u$ to $w$ and from $w$ to $u$. Let $\Delta C_1$ be the sum of all the corresponding cost differences when $w$ varies over $V(G) \setminus\left\{ u \right\}$. The part of the term corresponding to the bought links is $2(n-1) \alpha$. The part of the term corresponding to the sum of distances is: 

$$ \sum_{w \in V(G)\setminus \left\{u\right\}}  \sum_{x \in V(G)} d_{G+uw}(u,x)-d_G(u,x)+d_{G+uw}(w,x)-d_G(w,x)= $$

$$ =\sum_{x \in V(G)} \sum_{w \in V(G)\setminus \left\{u\right\}} d_{G+uw}(u,x)-d_G(u,x)+d_{G+uw}(w,x)-d_G(w,x) $$

 Where $G+uw$ denotes the graph $G$ together with the edge $uw$. Notice that the expression $d_{G+uw}(u,x)-d_G(u,x)+d_{G+uw}(w,x)-d_G(w,x)$ is less than or equal $0$ for every $x,w$, but  if $u \in A_r(x)$ and $w \in A_s(x)$ for $r,s$ with $r-s>1$ or $r-s<-1$ then $d_{G+uw}(u,x)-d_G(u,x) \leq -1$ or $d_{G+wu}(w,x)-d_G(w,x) \leq -1$, respectively. In this way, the expression $d_{G+uw}(u,x)-d_G(u,x)+d_{G+wu}(w,x)-d_G(w,x)$ is strictly negative for all $w \not \in A_u(x)$, where $A_u(x) = \left\{w \mid |d_G(u,x)-d_G(w,x)| \leq 1 \right\}=A_{d_G(u,x)-1}(x)\cup A_{d_G(u,x)}(x)\cup A_{d_G(u,x)+1}(x)$.
 
 Therefore $\Delta C_1\leq 2\alpha(n-1)-\sum_{x \in V(G) }(n-|A_u(x)|) $. So that, if $G$ is a \NE then $\sum_{x \in V(G) }|A_u(x)| \geq n^2-2(n-1)\alpha$.

 Therefore there exists at least some $v\in V(G)$ such that $|A_u(v)| > n-2\alpha$.

 Now let $u,w \in V(G)$ be two nodes and consider the two deviations that consist in adding a link from $u$ to $w$ and from $w$ to $u$. Let $\Delta C_2$ be the sum of the two corresponding cost differences. We have that $ \Delta C_2 \leq 2 \alpha -(n-|M_1(u,w)|)$ where $M_1(u,w) = \left\{z \mid |d_G(z,u)- d_G(z,w)| \leq 1 \right\}$. Since if $G$ is an equilibrium, then $|M_1(u,w)| \geq  n-2\alpha$.

 Finally, let $w \neq v$. Since $|M_1(w,v)| \geq n-2\alpha $, using that $|Y \cap Z| \geq |Y|+|Z|-|X|$ for $Y,Z$ subsets of $X$, we get $|M_1(w,v)\cap A_u(v)| > n-4\alpha$. Now, let $r$ be such that $A_u(v) = A_{r-1}(v) \cup A_r(v) \cup A_{r+1}(v)$, that is, $r = d_G(u,v)$, and pick $z \in M_1(w,v) \cap A_u(v)$. If $z \in A_{r-1}(v)$ then clearly $d_G(z,w) \in \left\{ r-2,r-1,r\right\}$ because $z \in M_1(w,v)$, too. Likewise, if $z \in A_{r}(v)$ or $z \in A_{r+1}(v)$ then $d_G(z,w) \in \left\{r-1,r,r+1 \right\}$ or $d_G(z,w) \in \left\{r,r+1,r+2\right\}$, respectively. Therefore, $ M_1(w,v) \cap A_u(v) \subseteq A_{r-2}(w)\cup A_{r-1}(w) \cup A_r(w) \cup A_{r+1}(w) \cup A_{r+2}(w)$. In this way  $|A_{r-2}(w) \cup A_{r-1}(w)\cup A_{r}(w) \cup A_{r+1}(w) \cup A_{r+2}(w)| > n-4\alpha$ for every $w \neq v$, so now the conclusion follows easily.

\end{proof}

 Given an undirected graph $G = (V,E)$ we denote by $G^k$ the graph having the same vertex set $V$ and as edges the set $\left\{uv \mid 0 <d_G(u,v) \leq k\right\}$.

\begin{proposition} 
\label{proposition-distance-almost-uniform-2}
If $G$ is $(5,\epsilon)-$distance-almost-uniform then $G^4$ is $(2,\epsilon)-$distance-almost-uniform and has diameter $\lceil diam(G)/4 \rceil$.
\end{proposition}

\begin{proof}

 Notice that $\forall u,v$ if $d_{G}(u,v) = 4k+l$ with $k\geq 0$ and $0<l \leq 4$ then $d_{G^4}(u,v) = k+1$. Therefore, any generic distance value $d$ changes to $\lceil d/4 \rceil$. From here the fact that $diam(G^4) = \lceil diam(G)/4 \rceil$. 

 Next, let $r$ be the value for which $\max_{i=0}^{4} |A_{r+i}(u)| \geq n(1-\epsilon)$. Clearly the distances $r,r+1,r+2,r+3,r+4$ are $5$ consecutive values. Then when changing from $G$ to $G^4$ these values collapse to a set consisting of exactly two consecutive values. From here the conclusion. 

\end{proof}

The conjecture stated in 
 \cite{AlonDHKL14} saying that \emph{Distance-almost-uniform graphs have diameter $O(\log n)$} seems to be refuted  recently by 
M. Lavrov and P.-S. Lo. in \cite{LaLo:17}.
Furthermore, using the same techniques considered in \cite{Demaineetal:07} we have not been able to deduce an upper bound equal or better than $O(\log n)$ on the $PoA$ for the range $\alpha < n/C$ with $C > 4$. 

\section{Conclusions}

In the study of the upper range, among the new techniques we have introduced, we would like to highlight the coordinate systems. Such systems have served us as an analytical tool to be more precise when calculating and bounding the cost differences and, in particular, they have been crucial in order to show the final lower bound. An interesting open question is whether analogous results could be obtained for the max model using coordinates in order to enlarge the range $\alpha > 129$ for which the $PoA$ is known to be constant. 

Furthermore, when considering the lower range, the last result in section \ref{sec:so-and-ne} provides new insight about the core problem of upper bounding the $PoA$.  It seems that the structural property related with $\epsilon-$distance-almost-uniform graphs is fundamental in order to guarantee  the equilibrium of the  network.

%Furthermore, apart from pointing out a previously unknown relation between two related and fundamental models, --the sum classic network creation game  and the sum basic from \cite{AlonDHKL14}-- this result also provides new insight about the core problem of upper bounding the $PoA$.  It seems that the structural property related with $\epsilon-$distance-almost-uniform graphs is fundamental in order to guarantee that the network is in equilibrium. 
%Last but not least, we think that this result can help us to understand better what are \NE like too, since, until know, very few examples are known. 

\medskip

%\textbf{Acknowledgments} This work was partially
%supported by funds from the Spanish Ministry for Economy and Competitiveness (MINECO) and the European Union (FEDER funds) under grant COMMAS (ref. TIN2013-46181-C2-1-R), and SGR 2014 1034 (ALBCOM) of the Catalan government.

%\bibliographystyle{abbrv}
%\bibliography{bibliografia,Bib}
%

\end{document}